\documentclass{article}
\title{The Generalized Reed-Muller codes and the radical powers of a modular algebra}
\author{Harinaivo ANDRIATAHINY\\Mention: Mathématiques et informatique,\\Domaine: Sciences et Technologies,\\Université d'Antananarivo, Madagascar\\e-mail: harinaivo.andriatahiny@univ-antananarivo.mg}

\usepackage{amsmath,amssymb}
\usepackage{amscd,amsfonts}
\usepackage{amsthm}
\usepackage[english]{babel}
\usepackage[T1]{fontenc}
\usepackage[ansinew]{inputenc}
\swapnumbers
\theoremstyle{plain}
\newtheorem{thm}{Theorem}[section]
\newtheorem{prop}[thm]{Proposition}
\newtheorem{lem}[thm]{Lemma}
\newtheorem{cor}[thm]{Corollary}
\newtheorem{thbch}[thm]{Theorem (Berman-Charpin)}

\theoremstyle{definition}

\newtheorem{rem}[thm]{Remark}
\newtheorem{note}[thm]{Notations}

\DeclareMathOperator{\card}{Card}
\DeclareMathOperator{\rad}{Rad}

\begin{document}

\maketitle

\begin{abstract}
First, a new proof of Berman and Charpin's characterization of the Reed-Muller codes over the binary field or over an arbitrary prime field is presented. These codes are considered as the powers of the radical of a modular algebra. Secondly, the same method is used for the study of the Generalized Reed-Muller codes over a non prime field.
\end{abstract}
\section{Introduction}
S.D. Berman [4] showed that the binary Reed-Muller codes may be identified with the powers of the radical of the group algebra over the two elements field $\mathbb{F}_{2}$ of an elementary abelian 2-group. P. Charpin [6] gave a generalization of Berman's result for Reed-Muller codes over a prime field.\\
Many authors explored Berman's idea and gave another proofs of Berman's theorem (see [2],[9]).\\
Recently, I.N. Tumaikin [11][12] studied the connections between Basic Reed-Muller codes and the radical powers of the modular group algebra $\mathbb{F}_{q}[H]$ where $H$ is a multiplicative group isomorphic to the additive group of the field $ \mathbb{F}_{q}$ of order $q=p^r$ where $p$ is a prime number and $r$ is an integer. The index of nilpotency of the radical of $\mathbb{F}_{q}[H]$ is $r(p-1)+1$.\\
This paper presents an elementary proof of Berman and Charpin's characterization of the Reed-Muller codes by using a polynomial approach. The modular algebra $\displaystyle{\mathbb{F}_{p}[X_{1},\ldots,X_{m}]\;/\;(\;
X_{1}^{p}-1,\ldots,X_{m}^{p}-1\;)}$ where $m\geq 1$ is an integer is used to represent the ambient space of the codes. It is isomorphic to the group algebra $\mathbb{F}_{p}[\mathbb{F}_{p^m}]$. We utilize some properties of a linear basis of the modular algebra.\\
We study also the case of the Generalized Reed-Muller (GRM) codes over a non prime field $\mathbb{F}_{q}$ (with $r>1$). We consider the modular algebra
\begin{equation*}
A=\displaystyle{\mathbb{F}_{q}[X_{1},\ldots,X_{m}]\;/\;(\;
X_{1}^{q}-1,\ldots,X_{m}^{q}-1\;)}.
\end{equation*}
The index of nilpotency of the radical $M$ of $A$ is $m(q-1)+1$, so there are $m(q-1)+1$ non-zero powers of $M$ (with $M^{0}=A$). It is well-known that there are $m(q-1)+1$ non-zero Reed-Muller codes of length $q^m$ over $\mathbb{F}_{q}$. The main result is Theorem 6.1 which gives the GRM codes over $\mathbb{F}_{q}$ which are radical powers of $A$. We show that except for $M^{0},M$ and $M^{m(q-1)}$, none of the radical powers of $A$ is a GRM code over the non prime field $\mathbb{F}_{q}$.\\
The outline of the paper is as follows. In section 2, we give the definition and some properties of the GRM codes of length $q^m$ over a finite field $\mathbb{F}_{q}$ (see [5]). This section is also devoted to the modular algebra $A$.
In section 3, we give a proof for the links (see [1],[4] and [6]) between GRM codes over a prime field $\mathbb{F}_{p}$ and the radical powers of $
\displaystyle{\mathbb{F}_{p}[X_{1},\ldots,X_{m}]\;/\;(\;
X_{1}^{p}-1,\ldots,X_{m}^{p}-1\;)}$.
In section 4, we study the functions $(x-1)^{i}$, $0\leq i\leq q-1$, over a finite field $\mathbb{F}_{q}$. In section 5, we examine the case of the GRM codes over a non prime field. In section 6, an exemple is given.

\section{The modular algebra $A=\mathbb{F}_{q}[X_{1},\ldots,X_{m}] /(
X_{1}^{q}-1,\ldots,X_{m}^{q}-1)$ and the GRM codes}
Let $q=p^r$ with $p$ a prime number and $r\geq 1$ an integer. We consider the finite field $\mathbb{F}_{q}$ of order $q$.\\
Let $P(m,q)$ be the vector space of the reduced polynomials in $m$ variables over $\mathbb{F}_{q}$:
\begin{equation}\label{redpoly}
P(m,q):=\left\{P(Y_{1},\ldots,Y_{m})=\sum_{i_1=0}^{q-1}\cdots\sum_{i_m=0}^{q-1}u_{i_1\ldots i_m}Y_1^{i_1}\ldots Y_m^{i_m} \mid u_{i_1\ldots i_m}\in \mathbb{F}_{q}\right\}.
\end{equation}
The polynomial functions from $(\mathbb{F}_{q})^m$ to $\mathbb{F}_{q}$ are given by the polynomials of $P(m,q)$.\\
Let $\nu$ be an integer such that $0\leq \nu \leq m(q-1)$. Consider the subspace of $P(m,q)$ defined by
\begin{equation*}
P_{\nu}(m,q):=\left\{P(Y_{1},\ldots,Y_{m})\in P(m,q) \mid \deg(P(Y_{1},\ldots,Y_{m}))\leq \nu\right\}.
\end{equation*}

Consider the ideal $I=\left(X_{1}^{q}-1,\ldots,X_{m}^{q}-1\right)$ of the ring $\mathbb{F}_{q}[X_{1},\ldots,X_{m}]$.\\
Set $x_1=X_1+I,\ldots,x_m=X_m+I$. Then
\begin{equation}
A=\left\{\sum_{i_1=0}^{q-1}\cdots\sum_{i_m=0}^{q-1}a_{i_1\ldots i_m}x_1^{i_1}\ldots x_m^{i_m} \mid a_{i_1\ldots i_m}\in \mathbb{F}_{q}\right\}
\end{equation}
$A$ is a local ring with maximal ideal $M$ which is the radical of $A$.\\
Let $d$ be an integer such that $0\leq d \leq m(q-1)$. Consider the powers $M^d$ of $M$. A linear basis of $M^d$ over $\mathbb{F}_{q}$ is
\begin{equation}
B_{d}:=\left\{(x_{1}-1)^{i_{1}}\ldots (x_{m}-1)^{i_{m}} \mid 0\leq i_{1},\ldots,i_{m}\leq q-1 ,i_{1}+\ldots+ i_{m}\geq d\right\}
\end{equation}
We have the following ascending sequence of ideals:
\begin{equation}\label{seqideal}
\left\{0\right\}=M^{m(q-1)+1}\subset M^{m(q-1)}\subset \cdots \subset M^{2}\subset M\subset A
\end{equation}
Let us fix an order on the set of monomials
\begin{equation*}
\left\{x_{1}^{i_1}\ldots x_{m}^{i_m} \mid 0\leq i_{1},\ldots,i_{m}\leq q-1\right\}.
\end{equation*}
Then we have the following remark:
\begin{rem}\label{identification}
Each element $\sum_{i_1=0}^{q-1}\cdots\sum_{i_m=0}^{q-1}a_{i_1\ldots i_m}x_1^{i_1}\ldots x_m^{i_m}$ of $A$ can be identified with the vector $(a_{i_1\ldots i_m})_{0\leq i_{1},\ldots,i_{m}\leq q-1}$ of $\left(\mathbb{F}_{q}\right)^{q^m}$ and vice-versa. Hence the modular algebra $A$ is identified with $\left(\mathbb{F}_{q}\right)^{q^m}$.
\end{rem}
Let $\alpha$ be a primitive element of the finite field $\mathbb{F}_{q}$. It is clear that $\mathbb{F}_{q}=\left\{0,1,\alpha,\alpha^{2},\ldots,\alpha^{q-2}\right\}$.\\
Set
\begin{equation}\label{fieldelem}
\beta_{0}=0 \quad\text{and}\quad \beta_{i}=\alpha^{i-1} \quad\text{for}\quad 1\leq i\leq q-1.
\end{equation}
When considering $P(m,q)$ and $A$ as vector spaces over $\mathbb{F}_{q}$, we have the following isomorphism:
\begin{equation}\label{isom}
\begin{aligned}
\phi :\quad  P(m,q)&\longrightarrow A \\
        P(Y_{1},\ldots,Y_{m}) &\longmapsto \sum_{i_1=0}^{q-1}\cdots\sum_{i_m=0}^{q-1}P(\beta_{i_1},\ldots ,\beta_{i_m})x_1^{i_1}\ldots x_m^{i_m}
\end{aligned}
\end{equation}
The Generalized Reed-Muller code of length $q^m$ and of order $\nu$ ($0\leq \nu \leq m(q-1)$) over $\mathbb{F}_{q}$ is defined by
\begin{equation}\label{grm}
C_{\nu}(m,q):=\left\{(P(\beta_{i_1},\ldots ,\beta_{i_m}))_{0\leq i_{1},\ldots,i_{m}\leq q-1} \mid P(Y_{1},\ldots,Y_{m})\in P_{\nu}(m,q)\right\}.
\end{equation}
It is a subspace of $\left(\mathbb{F}_{q}\right)^{q^m}$ and we have the following ascending sequence:
\begin{equation}\label{seqgrm}
\left\{0\right\}\subset C_{0}(m,q)\subset C_{1}(m,q)\subset\cdots\subset C_{m(q-1)-1}(m,q)\subset C_{m(q-1)}(m,q)=\left(\mathbb{F}_{q}\right)^{q^m}
\end{equation}
We need the following notations:
\begin{note}
Set $\left[0,q-1\right]=\left\{0,1,2,\ldots ,q-1\right\}$,\\
$\underline{i}:=(i_{1},\ldots,i_{m})\in ([0,q-1])^{m}$,\\
$\left|\underline{i}\right|:=i_{1}+\ldots +i_{m}$,\\
$\underline{j}\leq\underline{i}$ if $j_{l}\leq i_{l}$ for all $l=1,2,\ldots ,m$ where $\underline{j}:=(j_{1},\ldots,j_{m})\in ([0,q-1])^{m}$,\\
$\underline{x}:=(x_{1},\ldots,x_{m})$,\\
$\underline{x}^{\underline{i}}:=x_1^{i_1}\ldots x_m^{i_m}$.
\end{note}
\noindent Consider the polynomial
\begin{equation}\label{jennpoly}
B_{\underline{i}}(\underline{x}):=(x_{1}-1)^{i_1}\ldots (x_{m}-1)^{i_m}.
\end{equation}
It is clear that
\begin{equation}\label{binome}
(x_{l}-1)^{i_l}=\sum_{j=0}^{i_l}(-1)^{i_{l}-j}\binom{i_{l}}{j}x_{l}^{j}
\end{equation}
for all $l=1,2,\ldots ,m$.
\begin{prop}\label{dimeq}
Considering the sequences (\ref{seqgrm})and (\ref{seqideal}), we have
\begin{equation*}
\dim_{\mathbb{F}_{q}}(M^d)=\dim_{\mathbb{F}_{q}}(C_{m(q-1)-d}(m,q))
\end{equation*}
for $0\leq d\leq m(q-1)$.
\end{prop}
\begin{proof}
Consider the set $E:=\left\{\underline{i}\in ([0,q-1])^{m}\mid \left|\underline{i}\right|\geq d\right\}$.\\
Since $B_{d}=\left\{B_{\underline{i}}(\underline{x})\mid\left|\underline{i}\right|\geq d\right\}$ is a basis of $M^d$, then $\dim_{\mathbb{F}_{q}}(M^d)=\card(E)$ where $\card(E)$ denotes the number of elements in the set $E$.\\
Consider the set $F:=\left\{\underline{i}\in ([0,q-1])^{m}\mid \left|\underline{i}\right|\leq m(q-1)-d\right\}$.\\We have $\dim_{\mathbb{F}_{q}}(C_{m(q-1)-d}(m,q))=\dim_{\mathbb{F}_{q}}(P_{m(q-1)-d}(m,q))=\card(F)$.\\
The mapping
\begin{equation*}
\begin{aligned}
\theta : ([0,q-1])^{m} &\longrightarrow ([0,q-1])^{m}\\
         (i_{1},\ldots,i_{m})&\longmapsto (q-1-i_{1},\ldots,q-1-i_{m})
\end{aligned}
\end{equation*}
is a bijection and the inverse mapping is $\theta^{-1}=\theta$.\\
Let $\underline{i}\in E$. Then $\left|\underline{i}\right|\geq d$, and $\left|\theta(\underline{i})\right|=m(q-1)-\left|\underline{i}\right|\leq m(q-1)-d$. Hence $\theta(\underline{i}) \in F$. And it follows that $\theta(E)\subseteq F$.\\
Conversely, let $\underline{i}\in F$. Then $\left|\underline{i}\right|\leq m(q-1)-d$, and $\left|\theta(\underline{i})\right|=m(q-1)-\left|\underline{i}\right|\geq d$. So $\theta(\underline{i}) \in E$. Note that $\theta(\theta(\underline{i}))=\underline{i}$. Thus $F\subseteq \theta(E)$.\\
Therefore, $F= \theta(E)$ and $\card(E)=\card(F)$.
\end{proof}
Let $\beta_{k}\in \mathbb{F}_{q}$. Consider the indicator function
\begin{equation}\label{indicfunc}
F_{\beta_{k}}(Y_{l})=1-(Y_{l}-\beta_{k})^{q-1}
\end{equation}
with $1\leq l\leq m$.\\
Then $F_{\beta_{k}}(Y_{l})\in P(m,q)$ and
\begin{equation*}
F_{\beta_{k}}(\beta_{j})=\begin{cases}
         1& \text{if $j=k$},\\
         0& \text{otherwise}.
\end{cases}
\end{equation*}
Consider the interpolation function
\begin{equation}\label{interpfunc}
H_{i_l}(Y_l):=\sum_{k=0}^{i_l}(-1)^{i_{l}-k}\binom{i_{l}}{k}F_{\beta_{k}}(Y_{l}).
\end{equation}
Then, we have $H_{i_l}(Y_l)\in P(m,q)$,
\begin{equation*}
H_{i_l}(\beta_{j})=\begin{cases}
        (-1)^{i_{l}-j}\binom{i_{l}}{j}&\text{if $0\leq j\leq i_l$},\\
        0&\text{if $i_l < j\leq q-1$}
\end{cases}
\end{equation*}
and
\begin{equation}\label{bin}
(x_{l}-1)^{i_l}=\sum_{j=0}^{i_l}H_{i_l}(\beta_{j})x_{l}^{j}.
\end{equation}
Set
\begin{equation}\label{prodinterpf}
H_{\underline{i}}(\underline{Y}):=\prod_{l=1}^{m}H_{i_l}(Y_l).
\end{equation}
Thus
\begin{equation}\label{degr}
\deg(H_{\underline{i}}(\underline{Y}))=\sum_{l=1}^{m}\deg(H_{i_l}(Y_l)).
\end{equation}
\begin{prop}\label{jenpol}
We have $H_{\underline{i}}(\underline{Y})=\phi^{-1}(B_{\underline{i}}(\underline{x}))$, where $\phi$ is the isomorphism defined in (\ref{isom}), i.e.
\begin{equation*}
B_{\underline{i}}(\underline{x})=\sum_{\underline{j}\leq \underline{i}}H_{\underline{i}}(\beta_{j_1},\ldots ,\beta_{j_m})\underline{x}^{\underline{j}}
\end{equation*}
\end{prop}
\begin{proof}
\begin{equation*}
\begin{aligned}
B_{\underline{i}}(\underline{x})&=\prod^m_{l=1}(x_{l}-1)^{i_{l}}\\
&=\prod^m_{l=1}(\sum^{i_{l}}_{j_{l}=0}H_{i_{l}}(\beta_{j_{l}})x_{l}^{j_{l}})\\
&=\sum_{\underline{j}\leq
\underline{i}}(\prod^m_{l=1}H_{i_{l}}(\beta_{j_{l}}))\underline{x}^{\underline{j}}\\
&=\sum_{\underline{j}\leq
\underline{i}}H_{\underline{i}}(\beta_{j_1},\ldots ,\beta_{j_m})\underline{x}^{\underline{j}}.
\end{aligned}
\end{equation*}
\end{proof}

\section{The GRM codes over the prime field $\mathbb{F}_{p}$ and the radical powers of   $\mathbb{F}_{p}[X_{1},...,X_{m}]\;/\;(\;X_{1}^{p}-1,...,X_{m}^{p}-1\;)$}
In this section, we consider the case $r=1$, i.e. $q=p$ a prime number and $\mathbb{F}_{q}=\mathbb{F}_{p}$ a prime field.\\
Let $\mathbb{F}_{p}=\left\{0,1,2,\ldots,p-1\right\}$ and set $\beta_{k}=k$ for all $k=0,1,\ldots,p-1$.\\
Let us study $(x-1)^{i}$ for $0\leq i\leq p-1$ over $\mathbb{F}_{p}$.\\
It is clear that
\begin{equation*}
\displaystyle{(x-1)^{i}\;=\;\sum^i_{j=0}(-1)^{i-j}\binom{i}{j}x^{j}}.
\end{equation*}
For $k\in \mathbb{F}_{p}$, according to (\ref{redpoly}), we have
\begin{equation*}
F_{k}(Y)=1-(Y-k)^{p-1}=-\prod_{j=0(j\neq k)}^{p-1}(Y-j) \in P(1,p)
\end{equation*}
and
\begin{equation*}
F_{k}(j)=\begin{cases}
         1& \text{if $j=k$},\\
         0& \text{otherwise}.
\end{cases}
\end{equation*}
Let us consider the interpolation function
\begin{equation}\label{interfunc}
H_{i}(Y):=\sum_{k=0}^{i}(-1)^{i-k}\binom{i}{k}F_{k}(Y)
\end{equation}
which is in $P(1,p)$.\\
Thus, $H_{i}(j)=(-1)^{i-j}\binom{i}{j}$\; for\; $0\leq j\leq i$\\
and
\begin{equation*}
\displaystyle{(x-1)^{i}\;=\;\sum^i_{j=0}H_{i}(j)x^{j}}.
\end{equation*}
Let $a_{i,d}$ be the elements of $\mathbb{F}_{p}$ defined by
\begin{equation}\label{coeffi}
a_{i,d}:=\sum^i_{j=0}(-1)^{i-j}\binom{i}{j}j^{d}
\end{equation}
where $0\leq i\leq p-1$ and $0\leq d\leq p-1$.
\begin{prop}\label{intfunc}
The interpolation function (\ref{interfunc}) satisfies the relation
\begin{equation*}
H_{i}(Y)=a_{i,0}-\sum_{d=0}^{p-1}a_{i,d}Y^{p-1-d}
\end{equation*}
\end{prop}
\begin{proof}
\begin{equation*}
\begin{aligned}
H_{i}(Y)&=\sum_{k=0}^{i}(-1)^{i-k}\binom{i}{k}F_{k}(Y)\\
&=\sum_{k=0}^{i}(-1)^{i-k}\binom{i}{k}\left[1-(Y-k)^{p-1}\right]\\
&=\sum_{k=0}^{i}(-1)^{i-k}\binom{i}{k}\left[1-\sum_{d=0}^{p-1}(-1)^{d}k^{d}\binom{p-1}{d}Y^{p-1-d}\right]
\end{aligned}
\end{equation*}
Since we consider a field of characteristic $p$, then
$\binom{p-1}{d}=(-1)^{d}\mbox{ mod } p$. It follows that
\begin{equation*}
\begin{aligned}
H_{i}(Y)&=\sum_{k=0}^{i}(-1)^{i-k}\binom{i}{k}\left[1-\sum_{d=0}^{p-1}k^{d}Y^{p-1-d}\right]\\
&=\sum_{k=0}^{i}(-1)^{i-k}\binom{i}{k}-\sum_{d=0}^{p-1}\left[\sum_{k=0}^{i}(-1)^{i-k}\binom{i}{k}k^{d}\right]Y^{p-1-d}
\end{aligned}
\end{equation*}
\end{proof}
\begin{prop}\label{coef}
The coefficients $a_{i,d}$ satisfy the following properties
\par(i)- $a_{i,d}=\displaystyle{i\sum_{k=0}^{d-1}\binom{d-1}{k}a_{i-1,k}}
\quad\text{for}\quad 1\leq d\leq p-1\quad \text{and}\quad 1\leq i\leq p-1$.
\par(ii)- $a_{i,d}=0 \quad \text{for}\quad 1\leq i\leq p-1 \quad\text{and}\quad 0\leq d\leq i-1$.
\par(iii)- $a_{i,i}\neq 0\quad \text{for all}\quad i=1,2,\ldots ,p-1$.
\end{prop}
\begin{proof}
(i)- For $d\geq 1$,
\begin{equation*}
\begin{aligned}
a_{i,d}&=\sum^i_{j=0}(-1)^{i-j}\binom{i}{j}j^{d}=\sum^i_{j=1}(-1)^{i-j}\binom{i}{j}j^{d}\\
&=\sum^i_{j=1}(-1)^{i-j}\cdot
j\cdot\binom{i}{j}j^{d-1}.
\end{aligned}
\end{equation*}
But
\begin{equation*}
\binom{i}{j}=\frac{i}{j}\binom{i-1}{j-1},
\end{equation*}
thus
\begin{equation*}
\begin{aligned}
a_{i,d}&=\sum^i_{j=1}(-1)^{i-j}\cdot j\cdot\frac{i}{j}\binom{i-1}{j-1}j^{d-1}=i\sum^i_{j=1}(-1)^{i-j}\binom{i-1}{j-1}j^{d-1}\\
&=i\sum^i_{j=1}(-1)^{i-1-(j-1)}\binom{i-1}{j-1}((j-1)+1)^{d-1}.
\end{aligned}
\end{equation*}
By using the relation
\begin{equation*}
((j-1)+1)^{d-1}=\sum_{k=0}^{d-1}\binom{d-1}{k}(j-1)^{k},
\end{equation*}
and setting $J=j-1$, we have
\begin{equation*}
\begin{aligned}
a_{i,d}&=i\sum_{k=0}^{d-1}\binom{d-1}{k}(\sum^i_{j=1}(-1)^{i-1-(j-1)}\binom{i-1}{j-1}(j-1)^{k})\\
&=i\sum^{d-1}_{k=0}\binom{d-1}{k}(\sum^{i-1}_{J=0}(-1)^{i-1-J}\binom{i-1}{J}J^{k}).
\end{aligned}
\end{equation*}
(ii) and (iii) can be easily proved by induction on $i$ and by using (i).
\end{proof}
\begin{prop}\label{degintf}
\begin{equation*}
\deg(H_{i}(Y))=p-1-i.
\end{equation*}
\end{prop}
\begin{proof}
It is obvious by Proposition \ref{intfunc} and Proposition \ref{coef}.
\end{proof}
\begin{rem}
An explicit expression of $H_{i}(Y)$ as a polynomial of degree $p-1-i$ is
\begin{equation*}
H_i(Y) = \alpha_i \prod_{j=1}^{p-1-i} (Y+j),
\end{equation*}
where $\alpha_i = - i! \mbox{ mod } p$.
\end{rem}
\begin{rem}
Polynomials $H_i(Y)$, $0 \leq i \leq p-1$, satisfy the backward recurrence relation
\begin{eqnarray*}
& & H_{p-1}(Y) = 1, \\
& & H_i(Y) = \frac{1}{i+1} (Y-i-1) H_{i+1}(Y) \quad, \quad 0 \leq i \leq p-2.
\end{eqnarray*}
\end{rem}
\noindent It is clear that in this section, the Proposition \ref{jenpol} become
\begin{equation}\label{jennpolyn}
B_{\underline{i}}(\underline{x})=\sum_{\underline{j}\leq \underline{i}}H_{\underline{i}}(\underline{j})\underline{x}^{\underline{j}}
\end{equation}
The following Theorem is well-known (see [4],[6]).
\begin{thbch}
Let $C_{\nu}(m,p)$ be the Reed-Muller code of length $p^m$ and of order $\nu$ ($0\leq \nu\leq m(p-1)$) over the prime field $\mathbb{F}_{p}$ and $M$ the radical of $\mathbb{F}_{p}[X_{1},...,X_{m}]\;/\;(\;X_{1}^{p}-1,...,X_{m}^{p}-1\;)$. Then
\begin{equation*}
C_{\nu}(m,p)=M^{m(p-1)-\nu}.
\end{equation*}

\end{thbch}
\begin{proof}
Set $d:=m(p-1)-\nu$. The set $B_{d}=\left\{B_{\underline{i}}(\underline{x}) \mid \left|\underline{i}\right|\geq d\right\}$ is a linear basis of $M^d$ over $\mathbb{F}_{p}$. Consider $B_{\underline{i}}(\underline{x})=\sum_{\underline{j}\leq \underline{i}}H_{\underline{i}}(\underline{j})\underline{x}^{\underline{j}} \in M^d$. By (\ref{degr}) and Proposition \ref{degintf}, we have $\deg(H_{\underline{i}}(\underline{Y}))=\sum_{l=1}^{m}p-1-i_{l}=m(p-1)-\left|\underline{i}\right|\leq m(p-1)-d=\nu$. It follows from Remark \ref{identification} and (\ref{grm}) that $B_{\underline{i}}(\underline{x}) \in C_{\nu}(m,p)$. Thus $M^d\subseteq C_{\nu}(m,p)$. Moreover, if we take $r=1$ in Proposition \ref{dimeq}, we have $\dim_{\mathbb{F}_{p}}(M^d)=\dim_{\mathbb{F}_{p}}(C_{\nu}(m,p))$.
\end{proof}

\section{Some properties of $(x-1)^i , 0\leq i\leq q-1$, over an arbitrary finite field $\mathbb{F}_{q}$}
In this section, we consider the finite field $\mathbb{F}_{q}$ where $q=p^r$ with $p$ a prime number and $r\geq 1$ an integer.\\
We have already seen from (\ref{binome}), (\ref{indicfunc}) and (\ref{interpfunc}) of Section 2 that
\begin{equation*}
(x-1)^{i}=\sum_{j=0}^{i}(-1)^{i-j}\binom{i}{j}x^{j} \quad,\quad  0\leq i\leq q-1 ,
\end{equation*}
and the interpolation function is
\begin{equation}\label{interpolf}
H_{i}(Y):=\sum_{k=0}^{i}(-1)^{i-k}\binom{i}{k}F_{\beta_{k}}(Y) \quad ,\quad 0\leq i\leq q-1.
\end{equation}
And we have
\begin{equation*}
(x-1)^{i}=\sum_{j=0}^{i}H_{i}(\beta_{j})x^{j} \quad ,\quad 0\leq i\leq q-1.
\end{equation*}
\begin{lem}\label{lemm}
\begin{equation*}
\binom{p^{r}-1}{d}=(-1)^{d} \mbox{ mod } p
\end{equation*}
where $p$ is a prime number, $r\geq 1$ an integer and $0\leq d\leq p^{r}-1$.
\end{lem}
\begin{proof}
It can be proved easily by induction on $d$.
\end{proof}
The following proposition is fundamental.
\begin{prop}\label{intfonction}
The interpolation function (\ref{interpolf}) satisfies the relation
\begin{equation*}
H_{i}(Y)=\displaystyle{\sum_{d=1}^{q-1}\alpha^{-d}\left[(-1)^i-(\alpha^{d}-1)^{i}\right]Y^{q-1-d}}
\end{equation*}
where $\alpha$ is a primitive element of $\mathbb{F}_{q}$ and $1\leq i\leq q-1$.
\end{prop}
\begin{proof}
We have
\begin{equation*}
\begin{aligned}
H_{i}(Y)&=\sum_{k=0}^{i}(-1)^{i-k}\binom{i}{k}F_{\beta_{k}}(Y)\\
&=\sum_{k=0}^{i}(-1)^{i-k}\binom{i}{k}\left[1-(Y-\beta_{k})^{q-1}\right]\\
&=(-1)^{i}(1-Y^{q-1})+\sum_{k=1}^{i}(-1)^{i-k}\binom{i}{k}\left[1-(Y-\beta_{k})^{q-1}\right]\\
&=(-1)^{i}(1-Y^{q-1})+\sum_{k=1}^{i}(-1)^{i-k}\binom{i}{k}-\sum_{k=1}^{i}(-1)^{i-k}\binom{i}{k}(Y-\beta_{k})^{q-1}.
\end{aligned}
\end{equation*}
Since
\begin{equation*}
(Y-\beta_{k})^{q-1}=\sum_{d=0}^{q-1}(-1)^{d}(\beta_{k})^{d}\binom{q-1}{d}Y^{q-1-d}
\end{equation*}
and by Lemma \ref{lemm},
we have
\begin{equation*}
(Y-\beta_{k})^{q-1}=\sum_{d=0}^{q-1}(\beta_{k})^{d}Y^{q-1-d}.
\end{equation*}
Thus, by (\ref{fieldelem}),
\begin{equation*}
\begin{aligned}
H_{i}(Y)&=(-1)^{i}(1-Y^{q-1})+\sum_{k=1}^{i}(-1)^{i-k}\binom{i}{k}\\
&-\sum_{k=1}^{i}(-1)^{i-k}\binom{i}{k}\left[\sum_{d=0}^{q-1}(\alpha^{k-1})^{d}Y^{q-1-d}\right]\\
&=(-1)^{i}(1-Y^{q-1})+\sum_{k=1}^{i}(-1)^{i-k}\binom{i}{k}\\
&-\sum_{d=0}^{q-1}(\sum_{k=1}^{i}(-1)^{i-k}\binom{i}{k}(\alpha^{k-1})^{d})Y^{q-1-d}.
\end{aligned}
\end{equation*}
Since
\begin{equation*}
\sum_{k=1}^{i}(-1)^{i-k}\binom{i}{k}=(-1)^{i+1}
\end{equation*}
then
\begin{equation*}
\begin{aligned}
H_{i}(Y)&=(-1)^{i}-(-1)^{i}Y^{q-1}-(-1)^{i}-(\sum_{k=1}^{i}(-1)^{i-k}\binom{i}{k})Y^{q-1}\\
&-\sum_{d=1}^{q-1}(\sum_{k=1}^{i}(-1)^{i-k}\binom{i}{k}(\alpha^{k-1})^{d})Y^{q-1-d}\\
&=-\sum_{d=1}^{q-1}(\sum_{k=1}^{i}(-1)^{i-k}\binom{i}{k}(\alpha^{k-1})^{d})Y^{q-1-d}\\
&=-\sum_{d=1}^{q-1}\alpha^{-d}(\sum_{k=1}^{i}(-1)^{i-k}\binom{i}{k}(\alpha^{k})^{d})Y^{q-1-d}\\
&=-\sum_{d=1}^{q-1}\alpha^{-d}\left[\sum_{k=0}^{i}(-1)^{i-k}\binom{i}{k}(\alpha^{k})^{d}-(-1)^{i}\right]Y^{q-1-d}\\
&=-\sum_{d=1}^{q-1}\alpha^{-d}\left[(\alpha^{d}-1)^{i}-(-1)^{i}\right]Y^{q-1-d}\\
&=\sum_{d=1}^{q-1}\alpha^{-d}\left[(-1)^i-(\alpha^{d}-1)^{i}\right]Y^{q-1-d}.
\end{aligned}
\end{equation*}
\end{proof}
\begin{cor}\label{composite}
\begin{enumerate}
\item $H_{1}(Y)=-\sum_{d=0}^{q-2}Y^d$.
\item $H_{2}(Y)=-\sum_{k=0}^{q-2}(2-\alpha^{-k})Y^k$.
\item $H_{q-1}(Y)=1$.
\end{enumerate}
\end{cor}
\begin{rem}\label{intzero}
$H_{0}(Y)=F_{\beta_{0}}(Y)=F_{0}(Y)=1-Y^{q-1}$.
\end{rem}
\noindent The next corollary is important to what follows.
\begin{cor}\label{key}
If $\mathbb{F}_{q}$ is a non prime, then we have
\begin{equation*}
\deg(H_{q-2}(Y))=q-2.
\end{equation*}
\end{cor}
\begin{proof}
In Proposition \ref{intfonction}, for $i=q-2$, the coefficient of $Y^{q-2}$ is\\ $\alpha^{-1}\left[(-1)^{q-2}-(\alpha-1)^{q-2}\right]$. If $(\alpha-1)^{q-2}=(-1)^{q-2}$, then $(\alpha-1)^{q-1}=(-1)^{q-2}(\alpha-1)$. Since $\alpha$ is a primitive element of $\mathbb{F}_{q}$, then $\alpha\neq 1$ and $(\alpha-1)^{q-1}=1$. Thus  $1=(-1)^{q-2}\alpha+(-1)^{q-1}$, and $(-1)^{q-2}\alpha=0$. hence, $\alpha =0$. This is a contradiction.
\end{proof}

\section{The GRM codes of length $q^m$ over a non prime field $\mathbb{F}_{q}$ and the powers of the radical of $A=\mathbb{F}_{q}[X_{1},\ldots,X_{m}]\;/\;(\;
X_{1}^{q}-1,\ldots,X_{m}^{q}-1\;)$}
From sections 2 and 4, we have our main theorem:
\begin{thm}
Let $C_{\nu}(m,q)$ be the Generalized Reed-Muller code of length $q^m$ ($m\geq 1$ an integer) and of order $\nu$ over a non prime field $\mathbb{F}_{q}$ and $M=\rad(A)$ where $A=\mathbb{F}_{q}[X_{1},\ldots,X_{m}]\;/\;(\;
X_{1}^{q}-1,\ldots,X_{m}^{q}-1\;)$.
Then
\par(i)- $M^{m(q-1)}=C_{0}(m,q)$, $M=C_{m(q-1)-1}(m,q)$ and $M^{0}=C_{m(q-1)}(m,q)$
\par(ii)- $M^{i}\neq C_{m(q-1)-i}(m,q)$ for all $i$ such that $2\leq i\leq m(q-1)-1$.
\end{thm}
\begin{proof}
Since $\mathbb{F}_{q}$ is a non prime field then $q\geq 4$.\\
(i)-(a)- $M^{m(q-1)}$ is linearly generated over $\mathbb{F}_{q}$ by $B_{(q-1,\ldots,q-1)}(\underline{x})=\check{1}$ the ``all one word''. By (\ref{degr}) and Corollary \ref{composite}, we have $\deg(H_{(q-1,\ldots,q-1)}(\underline{Y}))=0$. It follows from Proposition \ref{jenpol}, Remark \ref{identification} and (\ref{grm}) that $B_{(q-1,\ldots,q-1)}(\underline{x})\in C_{0}(m,q)$. Thus $M^{m(q-1)}\subseteq C_{0}(m,q)$. And by Proposition \ref{dimeq}, we have\\ $\dim_{\mathbb{F}_{q}}(M^{m(q-1)})=\dim_{\mathbb{F}_{q}}(C_{0}(m,q))$.\\
(b)- Consider $B_{\underline{i}}(\underline{x}):=(x_{1}-1)^{i_1}\ldots (x_{m}-1)^{i_m}\in M$. There is an integer $l$ such that $1\leq l\leq m$ and $i_{l}\geq 1$. By Proposition \ref{intfonction}, $\deg(H_{i_{l}}(Y))\leq q-2$. And we have $\deg(H_{i}(Y))\leq q-1$ for all $i\neq i_{l}$. So $\deg(H_{\underline{i}}(\underline{Y}))\leq q-2+(m-1)(q-1)=m(q-1)-1$. Thus $B_{\underline{i}}(\underline{x})\in C_{m(q-1)-1}$. This implies that $M\subseteq C_{m(q-1)-1}(m,q)$, and by Proposition \ref{dimeq}, the equality holds.\\
(c)- It is obvious because $C_{m(q-1)}(m,q)\subseteq A=M^{0}$ and the Proposition \ref{dimeq} give the result.\\
(ii)- Consider the following sequence:\\
$\left\{0\right\}\subset M^{m(q-1)}\subset M^{m(q-1)-1}\subset \cdots \subset M^{m(q-1)-(q-2)+1}\subset M^{m(q-1)-(q-2)}\\ \subset\cdots\subset M^{m(q-1)-2(q-2)+1}\subset\cdots\subset M^{m(q-1)-(m-1)(q-2)}\subset\cdots\\ \subset M^{m(q-1)-m(q-2)+1}\subset M^{m}\subset M^{m-1}\subset M^{m-2}\subset\cdots\subset M^{2}\subset M\subset A$.\\
For simplicity, let us proceed step by step:\\
\underline{Step one}:\\
- $M^{m(q-1)-1}$ is linearly generated over $\mathbb{F}_{q}$ by $B_{\underline{i}}(\underline{x})$ such that $\left|\underline{i}\right|\geq m(q-1)-1$.\\
Consider $B_{(q-2,q-1,\ldots,q-1)}(\underline{x})$ which is in $M^{m(q-1)-1}$. By Corollary \ref{composite}, Corollary \ref{key} and (\ref{degr}),
we have $\deg(H_{(q-2,q-1,\ldots,q-1)}(\underline{Y}))=q-2 > 1$ (for $q\geq 4$).\\
Thus $B_{(q-2,q-1,\ldots,q-1)}(\underline{x})\notin C_{1}(m,q)$. It follows that $M^{m(q-1)-1}\neq C_{1}(m,q)$.\\
- Since $q\geq 4$, then $M^{m(q-1)-1}\subseteq M^{m(q-1)-(q-2)+1}$.\\
Therefore, $B_{(q-2,q-1,\ldots,q-1)}(\underline{x})\in M^{m(q-1)-(q-2)+1}$ (*).\\
And since $\deg(H_{(q-2,q-1,\ldots,q-1)}(\underline{Y}))=q-2 > q-3$, then $B_{(q-2,q-1,\ldots,q-1)}(\underline{x})\notin C_{q-3}(m,q)$.\\
Hence $M^{m(q-1)-(q-2)+1}\neq C_{q-3}(m,q)$.\\
- It is clear by (\ref{seqgrm}) that $B_{(q-2,q-1,\ldots,q-1)}(\underline{x})\notin C_{i}(m,q)$ for $2\leq i\leq q-4$, then by (\ref{seqideal}) we have $M^{m(q-1)-i}\neq C_{i}(m,q)$ for $2\leq i\leq q-4$.\\
In particular, the statement is proved for the case $m=1$.\\
\underline{Step two}:\\
- $M^{m(q-1)-(q-2)}$ is linearly generated over $\mathbb{F}_{q}$ by $B_{\underline{i}}(\underline{x})$ such that $\left|\underline{i}\right|\geq m(q-1)-(q-2)$.\\
Consider $B_{(q-2,q-2,q-1,\ldots,q-1)}(\underline{x})$ which is in $M^{m(q-1)-(q-2)}$ by (*).\\
We have $\deg(H_{(q-2,q-2,q-1,\ldots,q-1)}(\underline{Y}))=2(q-2) > q-2$ (for $q\geq 4$).\\
Thus $B_{(q-2,q-2,q-1,\ldots,q-1)}(\underline{x})\notin C_{q-2}(m,q)$. So $M^{m(q-1)-(q-2)}\neq C_{q-2}(m,q)$.\\
- Since $q\geq 4$, then $M^{m(q-1)-(q-2)}\subseteq M^{m(q-1)-2(q-2)+1}$.\\
Therefore $B_{(q-2,q-2,q-1,\ldots,q-1)}(\underline{x})\in M^{m(q-1)-2(q-2)+1}$.\\
And since $\deg(H_{(q-2,q-2,q-1,\ldots,q-1)}(\underline{Y}))=2(q-2) > 2(q-2)-1$, we have\\ $B_{(q-2,q-2,q-1,\ldots,q-1)}(\underline{x})\notin C_{2(q-2)-1}(m,q)$.\\
Hence $M^{m(q-1)-2(q-2)+1}\neq C_{2(q-2)-1}(m,q)$.\\
- For $q>4$, since $B_{(q-2,q-2,q-1,\ldots,q-1)}(\underline{x})\notin C_{i}(m,q)$ where $q-1\leq i\leq 2(q-2)-2$, then $M^{m(q-1)-i}\neq C_{i}(m,q)$ for $q-1\leq i\leq 2(q-2)-2$.\\
Continuing in this way, we apply the same method for each step. Thus, for the m-th step, we have\\
\underline{Step m}:\\
- $M^{m(q-1)-(m-1)(q-2)}$ is linearly generated over $\mathbb{F}_{q}$ by $B_{\underline{i}}(\underline{x})$ such that $\left|\underline{i}\right|\geq m(q-1)-(m-1)(q-2)$.\\
Consider $B_{(q-2,\ldots,q-2)}(\underline{x})$ which is in $M^{m(q-1)-(m-1)(q-2)}$ (for $m\geq 2$). By Corollary \ref{key} and (\ref{degr}),
we have $\deg(H_{(q-2,\ldots,q-2)}(\underline{Y}))=m(q-2) > (m-1)(q-2)$ (for $q\geq 4$).\\
Thus $B_{(q-2,\ldots,q-2)}(\underline{x})\notin C_{(m-1)(q-2)}(m,q)$. Therefore\\ $M^{m(q-1)-(m-1)(q-2)}\neq C_{(m-1)(q-2)}(m,q)$.\\
- Since $q\geq 4$, then $M^{m(q-1)-(m-1)(q-2)}\subseteq M^{m(q-1)-m(q-2)+1}$.\\
Hence $B_{(q-2,\ldots,q-2)}(\underline{x})\in M^{m(q-1)-m(q-2)+1}$.\\
And since $\deg(H_{(q-2,\ldots,q-2)}(\underline{Y}))=m(q-2) > m(q-2)-1$,\\then $B_{(q-2,\ldots,q-2)}(\underline{x})\notin C_{m(q-2)-1}(m,q)$.\\
Therefore $M^{m(q-1)-m(q-2)+1}\neq C_{m(q-2)-1}(m,q)$.\\
- For $q>4$, since $B_{(q-2,\ldots,q-2)}(\underline{x})\notin C_{i}(m,q)$ where $(m-1)(q-2)+1\leq i\leq m(q-2)-2$, we have $M^{m(q-1)-i}\neq C_{i}(m,q)$ for $(m-1)(q-2)+1\leq i\leq m(q-2)-2$.\\
\indent To end the proof, we consider the following final step:\\
- $M^{m(q-1)-m(q-2)}=M^{m}$ is linearly generated over $\mathbb{F}_{q}$ by $B_{\underline{i}}(\underline{x})$ such that $\left|\underline{i}\right|\geq m$.\\
Consider $B_{(0,2,1,\ldots,1)}(\underline{x})$ which is in $M^{m}$. By Corollary \ref{composite}, Remark \ref{intzero} and (\ref{degr})
we have $\deg(H_{(0,2,1,\ldots,1)}(\underline{Y}))=m(q-2)+1 > m(q-2)$.\\
Thus $B_{(0,2,1,\ldots,1)}(\underline{x})\notin C_{m(q-2)}(m,q)$. Hence $M^{m}\neq C_{m(q-2)}(m,q)$.\\
- $M^{m-1}$ is linearly generated over $\mathbb{F}_{q}$ by $B_{\underline{i}}(\underline{x})$ such that $\left|\underline{i}\right|\geq m-1$.\\
Consider $B_{(0,0,2,1,\ldots,1)}(\underline{x})$ which is in $M^{m-1}$.\\
We have $\deg(H_{(0,0,2,1,\ldots,1)}(\underline{Y}))=m(q-2)+2 > m(q-2)+1$.\\
Thus $B_{(0,0,2,1,\ldots,1)}(\underline{x})\notin C_{m(q-2)+1}(m,q)$. So $M^{m-1}\neq C_{m(q-2)+1}(m,q)$.\\
Similarly, we have finally:\\
- $M^{2}$ is linearly generated over $\mathbb{F}_{q}$ by $B_{\underline{i}}(\underline{x})$ such that $\left|\underline{i}\right|\geq 2$.\\
Consider $B_{(0,\ldots,0,2)}(\underline{x})$ which is in $M^{2}$.\\
We have $\deg(H_{(0,\ldots,0,2)}(\underline{Y}))=m(q-1)-1 > m(q-1)-2$.\\
Thus $B_{(0,\ldots,0,2)}(\underline{x})\notin C_{m(q-1)-2}(m,q)$. Hence $M^{2}\neq C_{m(q-1)-2}(m,q)$.
\end{proof}

\section{Example: the GRM codes of length $8^3$ over $\mathbb{F}_{8}$ and the powers of the radical of \\$A=\mathbb{F}_{8}[X_{1},X_{2},X_{3}]\;/\;(\;
X_{1}^{8}-1,X_{2}^{8}-1,X_{3}^{8}-1\;)$}
In this section, we consider the case $m=3$, $q=2^{3}=8$ and $n=8^{3}=512$.\\
Consider the modular algebra\\
$A=\mathbb{F}_{8}[X_{1},X_{2},X_{3}]\;/\;(\;X_{1}^{8}-1,X_{2}^{8}-1,X_{3}^{8}-1\;)\\=\left\{\sum_{i_1=0}^{7}\sum_{i_2=0}^{7}\sum_{i_3=0}^{7}a_{i_{1}i_{2}i_{3}}x_{1}^{i_1}x_{2}^{i_2}x_{3}^{i_3} \mid a_{i_{1}i_{2}i_{3}}\in \mathbb{F}_{8}\right\}$ with $x_1=X_1+I,x_2=X_2+I,x_3=X_3+I$ and $I=\left(X_{1}^{8}-1,X_{2}^{8}-1,X_{3}^{8}-1\right)$.\\
We have the sequence of ideals\\
$\left\{0\right\}\subset M^{21}\subset M^{20}\subset \cdots \subset M^{2}\subset M\subset A$ where $M=\rad(A)$.\\
Let $\nu$ be an integer such that $0\leq \nu\leq 21$. Consider the GRM codes $C_{\nu}(3,8)$ of length $512$ and of order $\nu$ over $\mathbb{F}_{8}$.\\
We have the ascending sequence:\\
$\left\{0\right\}\subset C_{0}(3,8)\subset C_{1}(3,8)\subset \cdots \subset  C_{19}(3,8)\subset  C_{20}(3,8)\subset  C_{21}(3,8)=(\mathbb{F}_{8})^{512}$.\\
-$M^{21}$ and $C_{0}(3,8)$ are linearly generated by $B_{(7,7,7)}(\underline{x})=\check{1}$\;(the ``all one word''), and we have $M^{21}=C_{0}(3,8)$.\\
-Since $B_{(6,7,7)}(\underline{x})\in M^{20}$ and $\deg(H_{(6,7,7)}(\underline{Y}))=6 > 1$, then \\$B_{(6,7,7)}(\underline{x})\notin C_{1}(3,8)$. Thus, $M^{20}\neq C_{1}(3,8)$.\\
Since $B_{(6,7,7)}(\underline{x})\in M^{i}$, $16\leq i\leq 19$ and $B_{(6,7,7)}(\underline{x})\notin C_{i}(3,8)$, $2\leq i\leq 5$, we have $M^{19}\neq C_{2}(3,8)$, $M^{18}\neq C_{3}(3,8)$, $M^{17}\neq C_{4}(3,8)$, $M^{16}\neq C_{5}(3,8)$.\\
-Since $B_{(6,6,7)}(\underline{x})\in M^{15}$ and $\deg(H_{(6,6,7)}(\underline{Y}))=12 > 6$, then\\ $B_{(6,6,7)}(\underline{x})\notin C_{6}(3,8)$. Therefore, $M^{15}\neq C_{6}(3,8)$.\\
Since $B_{(6,6,7)}(\underline{x})\in M^{i}$, $10\leq i\leq 14$ and $B_{(6,6,7)}(\underline{x})\notin C_{i}(3,8)$, $7\leq i\leq 11$, we have $M^{14}\neq C_{7}(3,8)$, $M^{13}\neq C_{8}(3,8)$, $M^{12}\neq C_{9}(3,8)$, $M^{11}\neq C_{10}(3,8)$, $M^{10}\neq C_{11}(3,8)$.\\
-Since $B_{(6,6,6)}(\underline{x})\in M^{9}$ and $\deg(H_{(6,6,6)}(\underline{Y}))=18 > 12$, then\\ $B_{(6,6,6)}(\underline{x})\notin C_{12}(3,8)$. This implies $M^{9}\neq C_{12}(3,8)$.\\
Since $B_{(6,6,6)}(\underline{x})\in M^{i}$, $4\leq i \leq 8$, and $B_{(6,6,6)}(\underline{x})\notin C_{i}(3,8)$, $13\leq i\leq 17$, we have $M^{8}\neq C_{13}(3,8)$, $M^{7}\neq C_{14}(3,8)$, $M^{6}\neq C_{15}(3,8)$, $M^{5}\neq C_{16}(3,8)$, $M^{4}\neq C_{17}(3,8)$.\\
-Since $B_{(0,2,1)}(\underline{x})\in M^{3}$ and $\deg(H_{(0,2,1)}(\underline{Y}))=19 > 18$, then\\ $B_{(0,2,1)}(\underline{x})\notin C_{18}(3,8)$. Thus, $M^{3}\neq C_{18}(3,8)$.\\
-We have $B_{(0,0,2)}(\underline{x})\in M^{2}$ and $\deg(H_{(0,0,2)}(\underline{Y}))=20 > 19$. Therefore,\\ $B_{(0,0,2)}(\underline{x})\notin C_{19}(3,8)$ and $M^{2}\neq C_{19}(3,8)$.

\end{document}